\documentclass[conference]{IEEEtran}

\IEEEoverridecommandlockouts
\IEEEpubid{\makebox[\columnwidth]{978-1-6654-4895-6/21/\$31.00~
\copyright2021 IEEE \hfill} \hspace{\columnsep}\makebox[\columnwidth]{ }}

\usepackage[utf8]{inputenc}
\usepackage{imakeidx}
\usepackage{amsmath, amsfonts, amssymb, stmaryrd, amsthm}
\usepackage{tikz-cd}
\usepackage{thmtools, thm-restate}

\theoremstyle{plain}
\newtheorem{theorem}{Theorem}
\newtheorem{lemma}{Lemma}[section]

\newtheorem{proposition}{Proposition}[section]
\theoremstyle{definition}

\newtheorem{example}{Example}[section]

\theoremstyle{remark}

\numberwithin{theorem}{section}

\newcommand{\ie}{\textit{i.e.}~}
\newcommand{\eg}{e.g.~}
\newcommand{\Fraisse}{Fra\"{i}ss\'{e}}

\newcommand{\struct}[1]{\mathcal{#1}}
\newcommand{\As}{\struct{A}}
\newcommand{\Bs}{\struct{B}}
\newcommand{\Cs}{\struct{C}}
\newcommand{\Rs}{\struct{R}}
\newcommand{\Ps}{\struct{P}}

\newcommand{\preford}{\sqsubseteq}

\newcommand{\Geg}{\hat{\mathbb{G}}^{\gd}}
\newcommand{\Gkdgx}[2]{\Geg_{#1,#2}}
\newcommand{\Gkdg}{\Gkdgx{k}{d}}
\newcommand{\Gdg}{\Gkdgx{\infty}{d}}

\newcommand{\clique}{\mathsf{clique}}
\newcommand{\cliquen}[1]{\clique_{#1}}

\newcommand{\Gk}{\mathbb{G}_{k}}

\newcommand{\rarr}{\rightarrow}
\newcommand{\id}{\mathsf{id}}
\newcommand{\last}{\lambda}

\newcommand{\sg}{\sigma}
\newcommand{\va}{\vec{a}}
\newcommand{\vb}{\vec{b}}
\newcommand{\vx}{\vec{x}}
\newcommand{\vy}{\vec{y}}

\newcommand{\RA}{R^{\As}}
\newcommand{\RGA}{R^{\GG \As}}
\newcommand{\RB}{R^{\Bs}}

\newcommand{\ve}{\varepsilon}
\newcommand{\vempty}{\varnothing}
\newcommand{\vphi}{\varphi}
\newcommand{\gsim}{\preceq^{\gd}}
\newcommand{\gsimk}{\preceq^{\gd}_{k}}
\newcommand{\gsimd}{\preceq^{\gd}_{\infty,d}}
\newcommand{\gbsim}{\sim^{\gd}}
\newcommand{\gbsimk}{\sim^{\gd}_{k}}
\newcommand{\gbsimd}{\sim^{\gd}_{\infty,d}}

\newcommand{\pref}{\sqsubseteq}

\newcommand{\downset}{{\downarrow}}

\newcommand{\IMP}{\; \Rightarrow \;}
\newcommand{\AND}{\; \wedge \;}

\newcommand{\CS}{\CStruct{\sg}}
\newcommand{\CStruct}[1]{\mathsf{Struct}(#1)}

\newcommand{\gd}{\mathfrak{g}}
\newcommand{\Gg}{\mathbb{G}^{\gd}}
\newcommand{\Gkg}{\mathbb{G}_{k}^{\gd}}
\newcommand{\GG}{\mathbb{G}}
\newcommand{\lb}{\llparenthesis}
\newcommand{\rb}{\rrparenthesis}
\newcommand{\epsA}{\ve_{\As}}
\newcommand{\dg}{\dagger}
\newcommand{\lt}{\langle}
\newcommand{\rt}{\rangle}

\newcommand{\MaxGuarded}{\mathsf{MaxGuarded}}
\newcommand{\EMG}{\CS^{\GG}}
\newcommand{\EMGk}{\CS^{\Gk}}
\newcommand{\lrarr}{\leftrightarrow}
\newcommand{\meet}{\sqcap}
\newcommand{\dc}{\smile}
\newcommand{\cvr}{\prec}

\newcommand{\hgraph}{\mathsf{HGraph}}
\newcommand{\hver}[1]{V_{#1}}
\newcommand{\hedg}[1]{E_{#1}}
\newcommand{\hyp}[1]{(\hver{#1},\hedg{#1})}
\newcommand{\HH}{\mathbb{H}}
\newcommand{\Hk}{\HH_{k}}
\newcommand{\until}{\mathbin{\mathbf{until}}}

\newcommand{\Guarded}{\mathsf{Guarded}^{\gd}}

\newcommand{\plays}{\mathsf{Plays}^{\gd}}
\newcommand{\Pt}{P_{\tau}}
\newcommand{\lbl}{\last}
\newcommand{\hp}{h^{+}}
\newcommand{\emb}{\rightarrowtail}

\title{Comonadic semantics for guarded fragments}
\author{
 \IEEEauthorblockN{Samson Abramsky}
 \IEEEauthorblockA{University of Oxford\\Email: samson.abramsky@cs.ox.ac.uk} 
 \and 
 \IEEEauthorblockN{Dan Marsden}
 \IEEEauthorblockA{University of Oxford\\Email: daniel.marsden@cs.ox.ac.uk}
 }

\date{April 2021}

\begin{document}

\maketitle

\begin{abstract}
    In previous work (\cite{abramsky2017pebbling,DBLP:conf/csl/AbramskyS18,AbramskyShah2020}), it has been shown how a range of model comparison games which play a central role in finite model theory, including Ehrenfeucht-\Fraisse, pebbling, and bisimulation games, can be captured in terms of resource-indexed comonads on the category of relational structures. Moreover, the coalgebras for these comonads capture important combinatorial parameters such as tree-width and tree-depth.

The present paper extends this analysis to quantifier-guarded fragments of first-order logic. We give a systematic account, covering atomic, loose and clique guards. In each case, we show that coKleisli morphisms capture winning strategies for Duplicator in the existential guarded bisimulation game, while back-and-forth bisimulation, and hence equivalence in the full guarded fragment, is captured by spans of open morphisms.
We study the coalgebras for these comonads, and show that they correspond to guarded tree decompositions.
We relate these constructions to a syntax-free setting, with a comonad on the category of hypergraphs.
\end{abstract}

\section{Introduction}

This paper  builds on the work initiated in \cite{abramsky2017pebbling}, and further developed in \cite{DBLP:conf/csl/AbramskyS18,AbramskyShah2020}, which applies methods from categorical semantics to finite model theory, combinatorics, and descriptive complexity.
The aim, as articulated in \cite{DBLP:conf/csl/AbramskyS18,AbramskyShah2020}, is to \emph{relate structure to power}, \ie to apply the structural and compositional methods of semantics to questions of expressiveness, resource bounds, and algorithmic complexity. 
The underlying motivation is to explore the new possibilities which arise
from entangling these very different, and hitherto almost disjoint fields of theoretical computer science.

In \cite{abramsky2017pebbling,DBLP:conf/csl/AbramskyS18} it was shown how a range of model comparison games which play a central role in finite model theory, including Ehrenfeucht-\Fraisse, pebbling, and bisimulation games, can be captured in terms of resource-indexed comonads on the category of relational structures. Moreover, the coalgebras for these comonads capture important combinatorial parameters such as tree-width and tree-depth.
In this paper, we extend this analysis to quantifier-guarded fragments of first-order logic \cite{andreka1998modal,gradel1999decision,gradel2014freedoms}. These guarded fragments have proved to be a fruitful generalization of the modal fragment, inheriting its good computational properties, such as decidability and the tree model property, while significantly enhancing its expressive power. The basic modal fragment was studied from the comonadic point of view in \cite{DBLP:conf/csl/AbramskyS18}. The extension to richer guarded fragments which we undertake here involves significant new ideas. From the combinatorial point of view, it involves the generalization from graphs to hypergraphs.
This shift to a higher-dimensional setting leads to new subtleties in the comonad constructions, manifested e.g.~in the use of quotients, and to coalgebraic descriptions of acyclicity of hypergraphs. 

In addition to these technical refinements, it is important to consider the guarded fragments and other significant logics to deepen our understanding of the scope and possibilities for the comonadic approach. The study of the guarded fragments brings several new features to light:
\begin{itemize}
    \item The clique guarded fragment presents obstacles to a resource-bounded version, which are addressed functorially in a novel fashion. This is discussed further below.
    \item The hypergraph comonad is a syntax-free version of the guarded constructions which is in several ways smoother than its relational counterparts, and may point the way to further constructions which are not tied to relational signatures.
    \item The general comonadic approach to back-and-forth equivalences initiated in \cite{AbramskyShah2020} is shown to apply to this new, more technically demanding setting, suggesting the possibility of a more general axiomatic approach.
\end{itemize}

We shall consider the three main notions of quantifier guarding which have appeared in the literature: atom, loose and clique guards. A fourth, packed guards, is equivalent to clique guards \cite{HodkinsonOtto2003}. On the one hand, to a large extent we are able to give a uniform treatment of these notions. It is also the case, though, that a significant difference arises when we consider resource-indexed versions of the guarded comonads. In particular, the clique guards involve  existential witnesses, which causes problems for a functorial approach to resource bounds. We shall find an answer to this problem by expanding the signature in such a way that clique guarding is reduced to atom guarding.

\section{Background}

We shall assume some knowledge of the most basic elements of category theory: categories, functors and natural transformations. Accessible references include \cite{pierce1991basic,abramsky2010introduction}.
Other notions will be introduced as needed.

A relational vocabulary $\sg$ is a set of relation symbols $R$, each with a specified positive integer arity.
A $\sg$-structure $\As$ is given by a set $A$, the universe of the structure, and for each $R$ in $\sg$ with arity $n$, a relation $\RA \subseteq A^n$. A homomorphism $h : \As \rarr \Bs$ is a function $h : A \rarr B$ such that, for each relation symbol $R$ of arity $n$ in $\sg$, for all $a_1, \ldots , a_n$ in $A$:
$\RA(a_1,\ldots , a_n) \IMP \RB(h(a_1), \ldots , h(a_n))$. We write $\CS$ for the category of $\sg$-structures and homomorphisms.

We recall that the Gaifman graph of a structure $\As$ is a reflexive graph with vertices $A$, such that two distinct elements are adjacent if they both occur in some tuple $\va \in \RA$ for some relation symbol $R$ in $\sg$.

We shall need a few notions on posets. A \emph{forest order} is a poset $(F, {\leq})$ such that for each $x \in F$, the elements below it form a finite linear order. A \emph{tree} is a forest order with a least element (the root). Note that a tree has binary meets, written $s \wedge t$. Given elements $s, t \in T$, the \emph{unique path} between $s$ and $t$ is $[s \wedge t, s] \cup [s \wedge t, t]$, where we use interval notation with respect to the tree order: $[s, t] := \{ v \mid s \leq v \leq t \}$. We write $s \cvr t$ for the covering relation in a poset, which holds if $s < t$, and $[s,t] = \{ s,t \}$. Given a subset $S$ of a poset $P$, $\downset S := \{ x \in P \mid \exists y \in S. \, x \leq y \}$.

Note that a forest is a disjoint union of trees. Given elements $x$, $y$ of a forest $F$, we write $x \dc y$ if they have a lower bound in $F$. In this case, they are in a common subtree of $F$, and their meet $x \wedge y$ exists. The relation $\dc$ is an equivalence relation, which partitions $F$ into its maximal subtrees.

\subsection{Guarded fragments}
The quantifier-guarded fragments of first-order logic we shall consider contain the usual atomic formulas and are closed under the boolean connectives. They are also closed under the following restricted forms of quantification: if $\vphi(\vx,\vy)$ is a guarded formula, in which all variables occurring free are in $\vx, \vy$,  then so are $\exists \vx.\, G(\vx, \vy) \AND \vphi(\vx,\vy)$ and $\forall \vx. \, G(\vx, \vy) \rarr \vphi(\vx,\vy)$, where $G(\vx, \vy)$ is a \emph{guard}. 
We shall consider three fragments corresponding to increasingly liberal notions of guard:
\begin{LaTeXdescription}
    \item[Atom guarded:] $G(\vx,\vy)$ is an atomic formula in which all the variables in $\vx,\vy$ occur.
    \item[Loosely guarded:] $G(\vx,\vy)$ is a conjunction of atomic formulas, such that each pair of variables, one occurring in $\vx$, and the other in $\vx \cup \vy$, must occur in one of the atomic formulas.
    \item[Clique guarded:] For each strictly positive natural number~$n$, there is a positive existential formula~$\cliquen{n}$ such that for every~$\As$, $\As \models \cliquen{n}(a_1,\ldots,a_n)$ if and only if~$a_1,\ldots,a_n$ form a clique in the Gaifman graph. A clique guard~$G(\vx,\vy)$ is a formula of the form~$\cliquen{n}(\vx,\vy)$.
\end{LaTeXdescription}
As is most conspicuous in the case of clique guards, the point of these syntactic conditions is to ensure that, for any structure $\As$, $\As \models G(\va,\vb)$ implies that  the set of elements occurring in the tuples $\va, \vb$ forms a clique in the Gaifman graph of $\As$. See \cite{HodkinsonOtto2003} for further details.

Another notion of guard which has been considered in the literature is that of \emph{packed} guards \cite{marx2001tolerance}. These are in fact equivalent to clique guards, as observed in \cite{HodkinsonOtto2003}, so we shall not consider them separately.
\begin{example}[Guarded Formulae]
The standard translation of modal formulae (see \eg \cite{blackburn2002modal}) is contained within the atom guarded fragment. For example, $\Diamond P$ would be translated to a formula of the form:
\begin{equation*}
    \exists y.\, R(x,y) \wedge P(y)
\end{equation*}
Even the atom guarded fragment is more general than this, allowing for example:
\begin{equation*}
    \exists y.\, R(y,x) \wedge P(y) \quad\mbox{ and }\quad \exists y.\, (y = y) \wedge P(y)
\end{equation*}
which modally would correspond to backwards and global modalities respectively. Neither of these can be expressed within basic modal logic. We can of course also incorporate relations of arity greater than two, corresponding to polyadic modal logics.

The motivation for the loosely guarded fragment is to allow the introduction of more complex modalities. A standard example is a strict until modality, ~$\psi \until \varphi$, which can be translated to:
\begin{equation*}
    \exists y. \, x \leq y \wedge \varphi(y) \wedge \forall z.\, (x \leq z \wedge z < y) \rightarrow \psi(z)
\end{equation*}
The observation that each of these guards induces a clique in the Gaifman graph leads to the introduction of the more general clique guarded fragment~\cite{gradel1999decision}.
\end{example}

\subsection{Guarded sets}
The semantic counterpart of guards are \emph{guarded sets} in a structure $\As$. These will always be cliques in the Gaifman graph of $\As$.
Given a tuple $\va \in \As^k$, the support of $\va$ is the set of elements occurring in the tuple.
For each of the three fragments we are considering, we have the corresponding notion of guarded set, indexed by a resource bound $k$:
a set~$X$ is \emph{(atomic/loosely/clique) $k$-guarded} if there exists a tuple~$\va$ of length at most~$k$, and (atomic/loose/clique) guard $G$ such that~$\As \models G(\va)$, and $X$ is contained in the support of~$\va$. We say that~$X$ is an \emph{exactly guarded set} if it is equal to the support of~$\va$.
\begin{lemma}
\label{lem:guarded-set-preservation}
(Exactly) guarded sets of each type are preserved by homomorphisms. Guarded sets are closed under subsets.
\end{lemma}
\begin{proof}
Each of the formulae involved in the guarded conditions is existential positive, and therefore preserved by homomorphisms, meaning exactly guarded sets are preserved.
If~$X$ is contained in the support of~$\va$ and~$h$ is a homomorphism, then~$h(X)$ is contained in the support of~$h(\va)$.

The final part is immediate from the definition.
\end{proof}

\subsection{Guarded simulation and bisimulation}
For modal logics, it is well established that the correct notion of equivalence between models is given by bisimulation \cite{blackburn2002modal}. The canonical result here is the van Benthem Theorem \cite{van1977modal}, which characterizes modal logic as the bisimulation-invariant fragment of first-order logic. 
This picture extends smoothly to guarded logics:
they are the guarded bisimulation invariant fragments of first order logic~\cite{andreka1998modal}. They can also be extended with fixed point operators, for which there is an analog of the Janin-Walukiewicz theorem for the~$\mu$-calculus~\cite{JaninWalukiewicz1996}, which says that guarded fixed point logic is the bisimulation invariant part of guarded second order logic, the natural guarded variant of second order logic~\cite{GradelHirschOtto2002}.

It is convenient to phrase guarded bisimulation in terms of Spoiler-Duplicator games.
For each notion of guarding $\gd$ (atomic, loose or clique), and $\sg$-structures $\As$, $\Bs$, the $\gd$-guarded simulation game from $\As$ to $\Bs$
is described as follows:
\begin{itemize}
\item Round 0: We set $X_0 := \vempty$, $\vphi_0 := \vempty$.
\item Round $n+1$: Spoiler specifies a $\gd$-guarded set $X_{n+1}$ in $\As$. Duplicator must respond with a $\gd$-guarded set $Y_{n+1}$ in $\Bs$, and a partial homomorphism $\vphi_{n+1} : X_{n+1} \to Y_{n+1}$, such that $\vphi_{n+1} |_{X} = \vphi_{n} |_{X}$, where $X = X_{n+1} \cap X_{n}$.
\end{itemize}

As usual, Duplicator wins if he has a response at each round, otherwise Spoiler wins. 

The $\gd$-guarded bisimulation game is specified similarly, with additional requirements on Duplicator:
\begin{itemize}
\item Round 0: We set $X_0 := \vempty$, $Y_0 := \vempty$, $\vphi_0 := \vempty$.
\item Round $n+1$: Spoiler now has two options.
\begin{itemize}
\item Option 1: Spoiler specifies a $\gd$-guarded set $X_{n+1}$ in $\As$. Duplicator must respond with a $\gd$-guarded set $Y_{n+1}$ in $\Bs$, and a partial isomorphism $\vphi_{n+1} : X_{n+1} \to Y_{n+1}$, such that $\vphi_{n+1} |_{X} = \vphi_{n} |_{X}$, where $X = X_{n+1} \cap X_{n}$.
\item Option 2: Spoiler specifies a $\gd$-guarded set $Y_{n+1}$ in $\Bs$. Duplicator must respond with a $\gd$-guarded set $X_{n+1}$ in $\As$, and a partial isomorphism $\vphi_{n+1} : X_{n+1} \to Y_{n+1}$, such that $\vphi_{n+1}^{-1} |_{Y} = \vphi_{n}^{-1} |_{Y}$, where $Y = Y_{n+1} \cap Y_{n}$.
\end{itemize}
\end{itemize}

Again, Duplicator wins if he has a response at each round, otherwise Spoiler wins.
We write $\As \gsim \Bs$ if Duplicator has a winning strategy for the $\gd$-guarded simulation game from $\As$ to $\Bs$, and $\As \gsimk \Bs$ if Duplicator has a winning strategy for the  version of the game where moves are restricted to $k$-guarded sets. Similarly, we write $\As \gbsim \Bs$ and $\As \gbsimk \Bs$ for the corresponding notions for bisimulation.

The logical significance of these notions is given by the following standard result \cite{andreka1998modal,gradel1999decision}.
\begin{theorem}
\label{bisimlogth}
\begin{enumerate}
\item $\As \gbsim \Bs$ iff $\As$ and $\Bs$ satisfy the same $\gd$-guarded formulas.
\item $\As \gbsimk \Bs$ iff $\As$ and $\Bs$ satisfy the same $\gd$-guarded formulas of guarded quantifier width $\leq k$.
\item $\As \gsim \Bs$ iff every existential positive $\gd$-guarded formula satisfied by $\As$ is satisfied by $\Bs$.
\item $\As \gsimk \Bs$ iff every existential positive $\gd$-guarded formula of guarded quantifier width $\leq k$ satisfied by $\As$ is satisfied by $\Bs$.
\end{enumerate}
\end{theorem}

\section{Guarded comonads}
\label{sec:guarded-comonads}
We shall now define the comonads corresponding to our guarded fragments. For each notion of guarding $\gd$, and resource index $k>0$, we shall define comonads $\Gkg$ and $\Gg$ for the $k$-bounded and unrestricted cases respectively. The definitions will be given uniformly in $\gd$, and we shall simply write $\GG$ and $\Gk$ for the generic case.

We shall use the prefix order on finite sequences, written $s \preford t$, which is a forest order on non-empty sequences, \ie the set of prefixes of any sequence forms a finite linear order. If two sequences $s$ and $t$ have a non-empty meet $s \sqcap t$ in this order, there is a unique path of last elements of the sequences connecting $s$ to  $t$ via  $s \sqcap t$. More formally, writing $\last(s)$ for the last element of a non-empty sequence $s$, the unique path between $s$ and $t$ is
\[ \{ \last(u) \mid u \in [s \meet t, s] \cup [s \meet t, t] \} \]
where we use interval notation with respect to the prefix order.

Given a structure $\As$, we must define a new structure $\GG \As$. We build the universe of $\GG \As$ in two steps. Firstly, we form the set of \emph{focussed plays}. These are pairs $\langle p, a \rangle$, where $p$ is a \emph{play}, \ie a non-empty list $[U_1, \ldots, U_n]$ of guarded sets in $\As$, corresponding to a sequence of moves by Spoiler in the simulation game; and $a \in U_n$ is the \emph{focus}, \ie the current element under consideration. Then, in order to enforce the constraints $\vphi_{n+1} |_{X} = \vphi_{n} |_{X}$ on Duplicator moves in the simulation game, we quotient this set by an equivalence relation: $\langle p, a \rangle \sim \langle q, a'\rangle$ iff (i) $a = a'$, (ii) the greatest common prefix $p \sqcap q$ is non-empty, and (iii) $a$ is a member of every guarded set in the unique path between $p$ and $q$. 
We write $\lb p, a \rb$ for the equivalence class of $\langle p,a \rangle$. 

\textbf{Remark} We prefer considering general guarded sets, rather than restricting attention to exactly guarded sets, as this seems more natural. We shall return to this question in section~\ref{sec:depth-bounds}.

The universe of $\GG \As$ will be the set of equivalence classes $\lb p,a \rb$. For a $\sg$-relation $R$, we define 
$\RGA \, := \, \{ (\lb p, a_1 \rb, \ldots , \lb p, a_r \rb) \mid \RA(a_1, \ldots , a_r) \}$.
From the definition of the equivalence relation, it is clear that this is well-defined.
This gives the object part of the functor $\GG$. 

For the counit, we define the map $\ve_{\As} : \GG(A) \rarr A$ by $\ve_{\As}(\lb p,a\rb) = a$.
\begin{proposition}
$\ve_{\As}$ is well-defined, and a homomorphism.
\end{proposition}
\begin{proof}
It is clear that~$\ve_{\As}$ is well defined as every representative of a single equivalence class has the same second component.

Now for $\sg$-relation~$R$, assume~$\RGA ( \lb p_1, a_1 \rb, \ldots, \lb p_n, a_n \rb ) $. By the definition of relations in $\GG \As$, there must exists~$p$ such that, for~$1 \leq i \leq n$:
\begin{equation*}
    \lb p_i, a_i \rb = \lb p, a_i \rb
\end{equation*}
and~$\RA(a_1,\ldots,a_n)$. By the definition of~$\ve_{\As}$ we then have:
\begin{equation*}
    \RA(\ve_{\As}(\lb p_1, a_1\rb),\dots,\ve_{\As}(\lb p_n, a_n \rb))
\end{equation*}
and so~$\ve_{\As}$ is a homomorphism.
\end{proof}
We now characterise the  guarded sets in $\Gg \As$.
For each play $p$ in $\Gg \As$, we define the set
\begin{equation*}
  S_p \; := \;  \{ \lb p, a \rb \mid a \in \last(p)  \} .
\end{equation*}
\begin{restatable}{lemma}{GuardedSetProperties}
\label{lem:guarded-set-properties}
\begin{enumerate}
    \item For every play $p$, $S_p$ is a $\gd$-guarded set.
    \item Every clique in the Gaifman graph of $\Gg \As$, and in particular every $\gd$-guarded set, is a subset of $S_p$ for some play $p$.
    \item For every $p$, the counit $\ve_\As$ restricts to an embedding on $S_p$.
\end{enumerate}
\end{restatable}
\begin{proof}
$(1)$ and $(3)$. Given a set $\{ \lb p,a_1 \rb, \ldots , \lb p, a_n \rb \}$ we must have $\{ a_1, \ldots , a_n \} \subseteq \last(p)$, which is a $\gd$-guarded set in $\As$.
Thus such a set forms a $\gd$-guarded set in $\Gg \As$.
Moreover, the restriction of $\ve_{\As}$ to such sets is injective, and preserves and reflects relations.

For $(2)$, we prove that every clique in $\Gg \As$ must take this form by induction on the size of the clique. The base case for singletons is immediate.
For the inductive step, suppose we have a clique $\{ \lb p_1, a_1 \rb, \ldots , \lb p_n, a_n \rb, \lb p,a \rb \}$.
By the induction hypothesis, for some play $s$, for all $i = 1, \ldots , n$:
\begin{equation}
\lb s, a_i \rb = \lb p_i, a_i \rb
\end{equation}
Since the set is a clique, for each $i$, $\lb p_i, a_i \rb$ is adjacent to $\lb p,a \rb$ in the Gaifman graph of $\Gg \As$. This implies that there are plays $q_i$ such that $\lb p_i, a_i \rb = \lb q_i, a_i \rb$, and
$\lb p,a \rb = \lb q_i, a \rb$, or equivalently:
\begin{equation}
    \lb q_i, a \rb = \lb q_j, a \rb, \quad 1 \leq i,j \leq n .
\end{equation}
We must find a play $t$ such that, for all $i$:
\[ \lb s, a_i \rb = \lb t, a_i \rb, \qquad \lb q_i, a \rb = \lb t, a \rb . \]
Note that this implies in particular that
\[
\{ a_1, \ldots , a_n, a \} \subseteq \last(t)
\]
Since the plays $s \meet q_i$ are all below $s$, they are linearly ordered. Without loss of generality, we can assume
\[ s \meet q_1 \pref \cdots \pref s \meet q_n . \]
The plays $q_i \meet q_n$ are all below $q_n$, and hence linearly ordered. Let the least element of this order be $q_j \meet q_n$.
Since $s \meet q_n$ and $q_j \meet q_n$ are both below $q_n$, they are comparable. We consider two cases.

Case I: $s \meet q_n \pref q_j \meet q_n$. In this case, we define $t := q_j \meet q_n$. For each $i$, since $q_j \meet q_n \in [s \meet q_i, q_i]$, $\lb t, a_i \rb = \lb q_i, a_i \rb = \lb p_i, a_i \rb = \lb s, a_i \rb$. Also, $\lb t, a \rb = \lb q_n, a \rb = \lb q_i, a \rb$.

Case II: $q_j \meet q_n \pref s \meet q_n$. In this case, we define $t := s \meet q_n$. For each $i$, since $s \meet q_n \in [s \meet q_i, s]$, $\lb t, a_i \rb =  \lb s, a_i \rb$. Also, since $s \meet q_n \in [q_j \meet q_n, q_n]$, $\lb t, a \rb = \lb q_n, a \rb = \lb q_i, a \rb$.
\end{proof}

\begin{example}
Note that a consequence of the lemma is that the image of every clique in $\Gg \As$ under $\ve_\As$ must be a $\gd$-guarded set in $\As$.
In the case of graphs, with a single binary edge relation, then if $\gd$ is atom guarding, and $\As$ is any graph, $\Gg \As$ is triangle-free.
\end{example}
Now given a homomorphism $h : \GG(\As) \rarr \Bs$, we define its \emph{Kleisli coextension} $h^* : \GG(\As) \rarr \GG(\Bs)$.
If $p = [ U_1, \ldots , U_n]$, then $h^*(\lb p, a\rb) = \lb q, h(\lb p,a\rb)\rb$, where $q = [ V_1 , \ldots , V_n ]$, and 
\[ V_j = \{ h(\lb [ U_1, \ldots , U_j ], a_j \rb) \mid a_j \in U_j \}, \quad 1 \leq j \leq n . \]
\begin{restatable}{proposition}{CoKleisliWellDefined}
$h^*$ is well-defined, and a homomorphism.
\end{restatable}
\begin{proof}
To establish that the Kleisli coextension is uniquely defined, we assume that~$\lt p_1, a \rt$ and~$\lt p_2, a \rt$ are representatives of the same equivalence class, and aim to show the stated construction results in the same element of~$\GG \Bs$. As the assumed pairs are equivalent, their greatest common prefix~$p_1 \sqcap p_2$ is non-empty, and~$a$ appears in every elements of the plays from the common prefix to both~$p_1$ and~$p_2$. Therefore, for~$q$ satisfying either:
\begin{equation*}
    p_1 \sqcap p_2 \preford q \preford p_1 \quad\mbox{ or }\quad p_1 \sqcap p_2 \preford q \preford p_2
\end{equation*}
we have equality:
\begin{equation*}
    \lb p_1 \sqcap p_2, a \rb = \lb q, a \rb
\end{equation*}
and so under the same conditions
\begin{equation*}
    h (\lb p_1 \sqcap p_2, a \rb ) = h (\lb q, a \rb)
\end{equation*}
and so~$h (\lb p_1 \sqcap p_2, a \rb)$ appears in each of the sets:
\begin{equation*}
    \{ h(\lb q, a_q \rb) \mid a_q \in \last(q) \}
\end{equation*}
We also note that, trivially, ~$h(\lb p_1, a \rb) = h(\lb p_1 \sqcap p_2, a \rb) = h(\lb p_2, a \rb)$. If~$p = [U_1,\ldots,U_n]$, let~$p'$ denote the sequence~$[V_1,\ldots,V_n]$ with the~$V_j$ as defined above. We have shown that~$p'_1$ and~$p'_2$ have a common prefix, and that $h( \lb p_1 \sqcap p_2, a \rb )$ appears in each element of the plays from the common prefix to the tail of~$p'_1$ and~$p'_2$. Therefore~$h^*(\lb p_1, a \rb ) = h^*(\lb p_2, a \rb)$, and the Kleisli coextension is uniquely defined. 

We must also confirm the codomain of~$h^*$ is~$\Gg \Bs$. To do so, we must establish that the sets~$V_j$ are guarded sets in~$\Bs$. As guarded sets are preserved by homomorphisms by lemma~\ref{lem:guarded-set-preservation}, it is sufficient to show that sets of the form:
\begin{equation}
\label{eq:guard-check}
    \{ \lb q, a_q \rb \mid a_q \in \last(q) \}
\end{equation}
are guarded in~$\Gg \As$, which is part of lemma~\ref{lem:guarded-set-properties}. 
\end{proof}

\begin{restatable}{theorem}{GuardedComonad}
\label{thm:guarded-comonad}
The triple $(\GG, \ve, (\cdot)^*)$ is a comonad in Kleisli form. This means that the following equations are valid:
\[ \ve_{\As}^* = \id_{\GG \As}, \qquad \ve \circ f^* = f, \qquad (g \circ f^*)^* = g^* \circ f^* . \]
\end{restatable}
\begin{proof}
For the first claim, consider~$\lb [U_1,\ldots,U_n], a \rb \in \Gg \As$. There exists
\begin{equation*}
    [V_1,\ldots,V_n]   
\end{equation*}
such that
\begin{align*}
    &\ve_{\As}^*(\lb [U_1,\ldots,U_n], a \rb) \\&= \lb [V_1,\ldots,V_n], \ve_{\As}(\lb [U_1, \ldots, U_n], a \rb) \rb \\ &= \lb [V_1, \ldots, V_n], a \rb
\end{align*}
It is therefore sufficient to show that $U_j = V_j$ for~$1 \leq j \leq n$. We have:
\begin{equation*}
    V_j = \{ \epsilon \lb [U_1,\ldots,U_j], a_j \rb \mid a_j \in U_j \} = \{ a_j \mid a_j \in U_j \} = U_j
\end{equation*}
completing this part of the proof.

For the second claim, for~$\lb p, a \rb \in \Gg \As$, there exists~$q$ such that
\begin{equation*}
    \ve_{\As} \circ f^*(\lb p, a \rb) = \ve_{\As}(\lb q, f(\lb p, a \rb) \rb) = f(\lb p, a \rb)
\end{equation*}

For the final claim, again consider~$\lb [U_1,\ldots,U_n], a \rb \in \Gg \As$. There exist
\begin{equation*}
    V_1,\ldots,V_n \qquad W_1,\ldots,W_n \qquad X_1,\ldots,X_n
\end{equation*}
such that:
\begin{align*}
    &g^* \circ f^* \lb [U_1,\ldots,U_n], a \rb \\&= g^* \lb [V_1,\ldots,V_n], f(\lb [U_1,\ldots,U_n], a \rb) \rb\\
    &= \lb [W_1,\ldots,W_n], g(\lb [V_1,\ldots,V_n], f(\lb [U_1,\ldots, U_n], a \rb) \rb) \rb
\end{align*}
and
\begin{align*}
    &(g \circ f^*)^* (\lb [U_1,\ldots,U_n], a \rb) \\&= \lb [X_1,\ldots,X_n], g \circ f^*(\lb [U_1,\ldots,U_n], a\rb) \rb\\
    &= \lb [X_1,\ldots, X_n], g(\lb [V_1,\ldots,V_n], f(\lb [U_1,\ldots,U_n], a\rb) \rb) \rb
\end{align*}
Therefore, it is sufficient to show
\begin{equation*}
    [W_1,\ldots,W_n] = [X_1,\ldots,X_n]
\end{equation*}
For index~$1 \leq j \leq n$
\begin{align*}
    X_j &= \{ g \circ f^* (\lb [U_1,\ldots,U_j], a_u \rb) \mid a_u \in U_j \}\\
    &= \{ g(\lb [V_1,\ldots,V_j], f(\lb [U_1,\ldots,U_j], a_u \rb)\rb) \mid a_u \in U_j \}
\end{align*}
and
\begin{align*}
    W_j &= \{ g(\lb [V_1,\ldots,V_j], a_v \rb)  \mid a_v \in V_j \}\\
    &= \{ g(\lb [V_1,\ldots,V_j], f(\lb [U_1,\ldots,U_j], a_u \rb) \rb) \mid a_u \in U_j \}
\end{align*}
Therefore the lists have the same length and are equal pointwise.
\end{proof}

It is then standard  \cite{manes2012algebraic} that $\GG$ extends to a functor by $\GG f = (f \circ \epsilon)^*$; that $\ve$ is a natural transformation; and that if we define the comultiplication $\delta : \GG \Rightarrow \GG^2$ by $\delta_{\As} = \id_{\As}^*$, then $(\GG, \ve, \delta)$ is a comonad.

Explicitly, for every $\sigma$-structure $\As$, let $\delta_{\As}:\GG(\As) \rarr \GG(\GG(\As))$ be $t = \lb [ U_1,\dots,U_n ], a \rb \mapsto \lb [ T_1, \dots, T_n ], t \rb$ where $T_j = \{\lb [ U_1,\dots,U_j ], a_j \rb \mid a_j \in U_j\}$ for all $j = 1,\dots,n$.    

For~$\gd$ either atom or loose guards, for each $k>0$, we obtain resource-bounded variants of these comonads $\Gk$ by restricting to $k$-guarded sets of the appropriate type in forming the universe of $\Gk \As$. Interestingly, this grading does not go through for clique guards. This is essentially because lemma~\ref{lem:guarded-set-properties} does not restrict appropriately in the presence of existential quantifiers in guard formulae. We return to this question in Section~\ref{sec:atoms-and-cliques}.
\begin{theorem}
For~$\gd$ either atom or loose guards,
the triple $(\Gk^{\gd}, \ve, (\cdot)^*)$ is a comonad in Kleisli form for all $k>0$. Here the counit and coextension are the restrictions of the corresponding operations for the unbounded case.
\end{theorem}
\begin{proof}
This simply involves inspecting the previous proofs, and noting that bounding the size of the guarded sets causes no trouble.
\end{proof}

We now turn to the connection between coKleisli morphisms for these comonads, and winning strategies for Duplicator in guarded simulation games.

We fix $\sg$-structures $\As$, $\Bs$.
\begin{restatable}{theorem}{TheoremSimulation}
For each notion of guarding $\gd$, there is a bijective correspondence between:
\begin{enumerate}
    \item CoKleisli morphisms $\GG^{\gd} \As \to \Bs$.
    \item Winning strategies for Duplicator in the $\gd$-guarded simulation game from $\As$ to $\Bs$.
\end{enumerate}
Thus $\As \gsim \Bs$ iff there is a coKleisli morphism $\GG^{\gd} \As \to \Bs$.
\end{restatable}
\begin{proof}
The arguments in both directions work uniformly in the choice of guard type.

Firstly, given a coKleisli morphism of type~$h : \GG^{\gd} \As \to \Bs$ we construct a winning Duplicator strategy for the $\gd$-guarded simulation game inductively on the round number. In the first round assume that Spoiler plays guarded set~$U_1$. By lemma~\ref{lem:guarded-set-properties}, $\ve_{\As}$ yields an embedding of~$U_1$ onto~$\{ \lb [U_1], u \rb \mid u \in U_1 \}$. The image under~$h$ of this set is guarded by lemma~\ref{lem:guarded-set-preservation}, and so the composite yields the required homomorphism of type~$U_1 \to \Bs$. Now assume Spoiler has played~$[U_1,\ldots,U_n]$ in the first~$n$ rounds, and now plays~$U_{n + 1}$. By the induction hypothesis, we have a successful strategy for the first~$n$-rounds, with final move the image of~$\{ [U_1,\ldots,U_{n + 1}], u \rb \mid u \in U_{n + 1} \}$ under~$h$.
Again we observe that~$U_{n + 1}$ embeds as~$\{ \lb [U_1,\ldots,U_{n + 1}], u \rb \mid u \in U_{n + 1} \}$, and composing this with~$h$ gives us a homomorphism of type~$U_{n + 1} \to \Bs$. Finally, this homomorphism will be consistent with the choices in the previous round by the definition of the quotient.

For the other direction, assume we have a winning strategy for Duplicator in the guarded simulation game from~$\As$ to~$\Bs$. We define our mapping:
\begin{equation*}
    \lb p, u \rb \mapsto h(u)
\end{equation*}
where~$h$ is the homomorphism indicated by the Duplicator strategy's response to the sequence of Spoiler moves appearing in~$p$. We must confirm that this is a well defined homomorphism. Assume:
\begin{equation*}
    \lb p, u \rb = \lb q, u \rb
\end{equation*}
Then there exists~$p \sqcap q$ such that:
\begin{equation*}
    \lb p, u \rb = \lb p \sqcap q, u \rb = \lb q, u \rb
\end{equation*}
with~$u$ appearing in every element of the extension of~$p \sqcap q$ to~$p$ and~$q$. Therefore Duplicator's strategy must yield homomorphisms agreeing on this element along these plays, and so~$\lb p, u \rb$ and~$\lb q, u \rb$ will agree. So the mapping is independent of choice of representative. We must also verify that relations are preserved. All relations in~$\GG^{\gd}(\As)$ are of the form:
\begin{equation*}
    R^{\GG^{\gd}(\As)}(\lb p, u_1 \rb, \ldots, \lb p, u_n \rb )
\end{equation*}
As Duplicator's strategy yields a homomorphism on the last element of~$p$, our construction preserves relations of this form.

The two constructions are clearly inverse to each other.

\end{proof}
Once again, for atom and loose guards, we have corresponding resource-bounded results.
\begin{theorem}
If~$\gd$ is either atom or loose guards, for $k>0$, there is a bijective correspondence between:
\begin{enumerate}
    \item CoKleisli morphisms $\Gk^{\gd} \As \to \Bs$.
    \item Winning strategies for Duplicator in the $\gd$-guarded $k$-bounded simulation game from $\As$ to $\Bs$.
\end{enumerate}
Thus $\As \gsimk \Bs$ iff there is a coKleisli morphism $\Gk^{\gd} \As \to \Bs$.
\end{theorem}

By virtue of theorem~\ref{bisimlogth}, this yields a comonadic characterization of equivalence of structures modulo existential positive guarded formulas. Before extending this analysis to bisimulation and equivalence modulo all guarded formulas, we firstly study the coalgebras for the guarded comonads.

\section{Coalgebras}

Recall that a coalgebra for a comonad $(G, \varepsilon, \delta)$ is a morphism $\alpha : A \to G A$ such that the following diagrams commute:
\begin{center}
\begin{tikzcd}
A \ar[r, "\alpha"] \ar[rd, "\id_A"']
& G A \ar[d, "\ve_{A}"] \\
& A
\end{tikzcd}
$\qquad \qquad$
\begin{tikzcd}
A  \ar[r, "\alpha"] \ar[d, "\alpha"']
& G A \ar[d,  "\delta_{A}"] \\
G A  \ar[r, "G \alpha"] 
& G^2 A
\end{tikzcd}  
\end{center}
Note in particular that a $G$-coalgebra structure on $A$ makes it a retract of $G A$ via the counit $\epsA$. In our setting, this implies that coalgebra maps $\alpha : \As \to \GG \As$ are embeddings of $\sg$-structures. 
In \cite{abramsky2017pebbling,DBLP:conf/csl/AbramskyS18} it was shown that coalgebras for the Ehrenfeucht-\Fraisse, pebbling and modal comonads correspond to various forms of combinatorial decompositions of $\sg$-structures. This led to coalgebraic characterizations of important combinatorial parameters such as tree-depth and tree-width. We pursue a similar analysis here for the guarded comonads.

The appropriate notion of ``treelike decomposition'' of a structure in the guarded case is defined as follows. Let $\As$ be a $\sg$-structure.
We write $\Guarded(\As)$ for the set of $\gd$-guarded subsets of $\As$, and $\plays(\As)$ for the $\gd$-guarded plays, \ie the finite, non-empty sequences of $\gd$-guarded sets.
Given a function $\tau : A \to \plays(\As)$, let $\Pt$ be the image of $\tau$. 
As a sub-poset of $\plays(\As)$ under the prefix order, $\Pt$ is a forest. 
We say that $\tau$ is
a \emph{$\gd$-guarded decomposition of $\As$} if it satisfies the following conditions: 
\begin{enumerate}
\item \label{ax:reflex} It is \emph{reflexive}: For all $a \in A$, $a \in \lbl \circ \tau(a)$.

    \item \label{ax:relation-covering} It is~\emph{edge covering}; if $a$ and $b$ are adjacent in the Gaifman graph of $\As$,  then there exists~$p \in \Pt$ such that~$\{ a, b \} \subseteq \lbl(p)$. 
    \item It is \emph{minimal}: if $\tau(a) = p$, then for all $q \in \plays(\As)$, if $\lb p, a \rb\ = \lb q,a \rb$, then $p \preford q$.
    \item It is \emph{vertex connected}: for all $q \in \downset \Pt$, if $a \in \lbl(q)$, then $\lb \tau(a), a \rb = \lb q, a \rb$.
   \end{enumerate}


We derive some useful properties of decompositions.

\begin{lemma}\label{injlemm}
For $p, q \in \Pt$, if $\lambda(p) \subseteq \lambda(q)$, then $p \preford q$. In particular, $\lambda$ is injective on $\Pt$.
\end{lemma}
\begin{proof}
Suppose $\lambda(p) \subseteq \lambda(q)$, where $p = \tau(a)$ and $q = \tau(b)$. Then $a \in \lbl(p) \subseteq \lbl(q)$,  and by vertex connectedness, $\lb p, a \rb = \lb q, a \rb$.
By minimality,  $p \preford q$.
\end{proof}

\begin{lemma}\label{caplemm}
For $p, q \in \Pt$, if $a \in \lbl(p) \cap \lbl(q)$, then the meet $p \wedge q$ exists in $\Pt$, and $a \in \lbl(p \wedge q)$.
\end{lemma}
\begin{proof}
Using vertex connectedness twice, $\lb p,a \rb = \lb \tau(a),a \rb = \lb q,a \rb$, and by minimality, $\tau(a)$ is a lower bound of $p$ and $q$. Hence $p \wedge q$ exists, and $a \in \lbl(p \wedge q)$.
\end{proof}

\begin{restatable}{lemma}{GuardedSetCovering}
\label{GuardedSetCovering}
If~$\tau$ is a~$\gd$-guarded decomposition of a structure~$\As$, then every
clique in the Gaifman graph of $\As$ is contained in~$\lbl(p)$ for some~$p \in \Pt$.
\end{restatable}
\begin{proof}
We  argue by induction on the cardinality of the clique. The base case of a singleton is trivial.

Now assume~$C \cup \{ d \}$ is a clique. By the induction hypothesis on~$C$, there exists~$p \in \Pt$ such that $C \subseteq \lbl(p)$. We seek an element $q \in \Pt$ such that $C \cup \{d \} \subseteq \lbl(q)$.

By the clique condition and edge covering, for each~$a \in C$ we must have~$q^{a} \in \Pt$ such that
    $\{ a, d \} \subseteq \lbl(q^{a})$.
As each of the elements~$p \wedge q^{a}$ is below~$p$, they can be linearly ordered, say as
\[ p \wedge q^{a_1} \, \preford \cdots \preford \, p \wedge q^{a_n} . \]
Let $r := p \wedge q^{a_n}$. For each $i$, $a_i \in \lbl(q^{a_i}) \cap \lbl(p)$, so by vertex connectedness, $C \subseteq \lbl(r)$.

Now consider the elements $q^{a_i} \wedge q^{a_n}$, $i = 1, \ldots , n-1$. Since these are all below $q^{a_n}$, they are linearly ordered, say with least element $s := q^{a_j} \wedge q^{a_n}$.
Since $r$ and $s$ are both below $q^{a_n}$, they are comparable in the prefix order. We consider two cases:
\begin{enumerate}
    \item $r \preford s$. Since $d \in \lbl(q^{a_j}) \cap \lbl(q^{a_n})$, by vertex connectedness, $d \in \lbl(s)$. For each $i$, $a_i \in \lbl(p) \cap \lbl(q^{a_i})$, hence by vertex connectedness, $a_i \in \lbl(s)$. Thus we can take $q := s$.
    \item $s \preford r$. Since $d \in \lbl(q^{a_j}) \cap \lbl(q^{a_n})$, by vertex connectedness, $d \in \lbl(r)$. Thus we can take $q := r$.
\end{enumerate}
\end{proof}

As an immediate corollary of this result, we have
\begin{proposition}\label{cliquecorrprop}
If $\As$ has a $\gd$-guarded  decomposition $\tau$, then 
\begin{enumerate}
    \item  Every maximal $\gd$-guarded subset of $A$ is $\lbl(p)$ for some $p \in \Pt$.
    \item Every clique in the Gaifman graph of $\As$ is $\gd$-guarded.
\end{enumerate}
\end{proposition}

\begin{proposition}
If $\As$ has a $\gd$-guarded decomposition $\tau$, then $|A| \geq |\MaxGuarded(\As)|$, where $\MaxGuarded(\As)$ is the set of maximal $\gd$-guarded subsets of $A$.
\end{proposition}
\begin{proof}
By Proposition~\ref{cliquecorrprop}(1), $\lbl \circ \tau$ maps $A$ surjectively onto $\MaxGuarded(\As)$.
\end{proof}
\begin{example}
In the case of graphs, with a single binary edge relation, then if $\As$ is any graph with an atom-guarded decomposition, the number of edges of $A$ is less than or equal to the number of vertices. In fact, one can show that a simple graph has an atom-guarded decomposition iff it is acyclic, \ie a forest. Thus guarded decompositions generalise the unravelling construction from modal logic \cite{blackburn2002modal}.
\end{example}

It will be convenient to have an explicit formulation of the conditions on a homomorphism $\gamma : \As \to \Gg \As$ imposed by the two coalgebra diagrams.
\begin{lemma}\label{coalgdiagramlemm}
\begin{enumerate}
    \item The condition imposed by the first diagram is 
    that, for all $a \in A$, if $\gamma(a) = \lb p,b \rb$, then $b = a$.
    \item The condition imposed by the second diagram is that, if $\gamma(a) = \lb [U_1,\ldots, U_n],a \rb$, then for all $i$, $1 \leq i \leq n$, and for all $u \in U_i$, $\gamma(u) = \lb [U_1, \ldots , U_i], u \rb$.
\end{enumerate}
\end{lemma}
\begin{proof}
These follow directly by unpacking the definitions of the comonad constructions.
\end{proof}
We can use a $\gd$-guarded decomposition $\tau$ of $\As$ to define a function $\gamma : A \to \Gg A$. For each $a \in A$, we define
$\gamma(a) := \lb \tau(a), a \rb$.

\begin{proposition}\label{forestcoalgprop}
The function $\gamma$ is a $\Gg$-coalgebra.
\end{proposition}
\begin{proof}
Firstly, we show that $\gamma$ is a $\sg$-homomorphism. Given $\RA(a_1,\ldots,a_n)$, by Lemma~\ref{GuardedSetCovering}, $\{a_1,\ldots,a_n\} \subseteq \lbl(p)$ for some $p \in \Pt$. Let $p_i = \tau(a_i)$, $1 \leq i \leq n$. By vertex connectedness and minimality, all the $p_i$ are below $p$. Hence they are linearly ordered, with maximum some $p_j$. Now for each $i$, by vertex connectedness, $\gamma(a_i) = \lb p_i, a_i \rb = \lb p_{j}, a_i \rb$, and hence $R^{\Gg \As}(\lb p_{1}, a_1 \rb, \ldots , \lb p_{n}, a_n \rb)$ as required.

We verify the coalgebra diagrams in the form given in Lemma~\ref{coalgdiagramlemm}.
The first coalgebra diagram holds directly from the  reflexivity of $\tau$.
For the second diagram, 
if $u \in U_i$, then by vertex connectedness,
 $\gamma(u) = \lb \tau(u), u \rb = \lb [U_1, \ldots , U_i], u \rb$.
\end{proof}

To go from $\GG^{\gd}$-coalgebras to $\gd$-guarded  decompositions, we firstly  look more closely into the structure of $\sim$-equivalence classes in $\GG \As$.
\begin{lemma}
The elements of an equivalence class $\lb p, a \rb$ are tree-ordered by the prefix order on the first components.
\end{lemma}
\begin{proof}
The prefix order on plays is a forest order. By the definition of the $\sim$-equivalence relation, the plays in an equivalence class must form a subtree.
\end{proof}
Thus we can write $\lb p,a \rb^{\dg} = \langle q, a \rangle$, where $q \preford p$ is the least play in the equivalence class under the prefix ordering.

Given a coalgebra $\gamma : \As \to \GG^{\gd} \As$, we let $\tau(a) = p$, where $\gamma(a)^{\dg} = \langle p, a \rangle$.
This defines a map $\tau : A \to \plays(\As)$, from $A$ to $\gd$-guarded plays on $\As$.
\begin{proposition}\label{coalgforestprop}
$\tau$ is a $\gd$-guarded decomposition of $\As$.
\end{proposition}
\begin{proof}
Reflexivity follows immediately from the first coalgebra diagram, and minimality from the definition of $\tau$.

If $\RA(a_1, \ldots , a_n)$, then since $\gamma$ is a homomorphism, $R^{\Gg \As}(\gamma(a_1), \ldots ,  \gamma(a_n))$. This implies that for some play $p$, for all $i = 1, \ldots , n$,
$\lb p, a_i \rb = \lb p_i, a_i \rb$, where $p_i = \tau(a_i)$. The $p_i$ are below $p$, and hence form a chain, with maximum element some $p_j$. Then we have $\{ a_1, \ldots , a_n \} \subseteq \lambda(p_j)$, and the edge covering property is satisfied.

For vertex connectedness, if $q \preford p = \tau(b)$ and $a \in \lbl(q)$, then by the second coalgebra diagram, $\gamma(a) = \lb q,a \rb$, while by definition $\gamma(a) = \lb \tau(a), a \rb$.
Hence $\lb \tau(a), a \rb = \lb q,a \rb$ as required.
\end{proof}

\begin{theorem}
\label{Gcoalgcoverth}
Let $\As$ be a $\sg$-structure. There is a bijective correspondence between:
\begin{enumerate}
    \item $\GG^{\gd}$-coalgebras $\As \to \GG^{\gd} \As$
    \item  $\gd$-guarded  decompositions $\tau$ of $\As$.
\end{enumerate}
\end{theorem}
\begin{proof}
The transformations from decompositions to coalgebras and back given in Propositions~\ref{forestcoalgprop} and~\ref{coalgforestprop} are mutually inverse.
\end{proof}
We can adapt these notions to the resource-bounded case. A $\gd$-guarded  decomposition $\tau$ of $\As$ is $k$-bounded if for all $p \in \Pt$, $\lbl(p)$ has cardinality $\leq k$.
\begin{theorem}
Let $\As$ be a $\sg$-structure. There is a bijective correspondence between:
\begin{enumerate}
    \item $\Gk^{\gd}$-coalgebras $\As \to \Gk^{\gd} \As$
    \item $k$-bounded $\gd$-guarded  decompositions $\tau$ of $\As$.
\end{enumerate}
\end{theorem}
\begin{proof}
This is a straightforward restriction of theorem~\ref{Gcoalgcoverth}.
\end{proof}

\noindent We briefly illustrate how these decompositions give a new perspective on coalgebras for~$\GG$.
\begin{example}[Guarded Forest Decompositions and Coalgebras]
Let~$\sg$ be the signature with one binary relation~$R$. Consider the $\sg$-structure~$\As$ with universe~$\{a,b,c\}$, and relations~$R^\As(a,b)$ and~$R^\As(b,c)$.
There is a coalgebra~$\gamma : \As \rightarrow \GG(\As)$ with
\begin{align*}
    \gamma(a) &= \lb [\{ a, b \}], a \rb\\
    \gamma(b) &= \lb [\{ a, b \}], b \rb\\
    \gamma(c) &= \lb [\{ a, b \}, \{ b, c \}], c \rb
\end{align*}
The corresponding atom guarded tree decomposition has Hasse diagram:
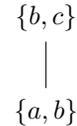
\[
\begin{tikzcd}[arrows=-]
\{ b, c \} \\
\{ a, b \} \uar
\end{tikzcd}
\]
In this case, we get another atom-guarded tree decomposition by inverting this Hasse diagram, and so a second coalgebra structure. The two coalgebras correspond to traversing the underlying graph from opposite ends.

For the same relational signature, the structure~$\Bs$ on~$\{ a, b, c\}$ with
$R^\Bs(a,b)$, $R^\Bs(b,c)$ and~$R^\Bs(c,a)$ cannot be formed into an atom guarded tree, and so does not carry a coalgebra structure with respect to atom guards.
If on the other hand, we consider loose guards, the set~$\{a, b, c\}$ is guarded by
\begin{equation*}
    R(a,b) \AND R(b,c) \AND R(c,a)
\end{equation*} 
and so $\Bs$ carries a trivial loosely-guarded decomposition, and corresponding coalgebra.
\end{example}
We define the \emph{$\gd$-guarded tree-width} of a structure $\As$ to be the least $k$ such that $\As$ has a  $k$-bounded $\gd$-guarded decomposition. The \emph{$\gd$-guarded coalgebra number} of $\As$ is the least $k$ such that there is a coalgebra $\As \to \Gk^{\gd} \As$. As an immediate consequence of the previous theorem, we obtain the following result.
\begin{theorem}
For any structure $\As$, the $\gd$-guarded tree-width of $\As$ and the $\gd$-guarded coalgebra number of $\As$ coincide.
\end{theorem}

We now consider morphisms. Recall that a morphism of $\GG$-coalgebras $h : (\As, \alpha) \to (\Bs, \beta)$ is a  $\sg$-homomorphism $h : \As \to \Bs$ such that the following diagram commutes:
\[ \begin{tikzcd}
\As \ar[d, "h"'] \ar[r, "\alpha"] & \GG \As \ar[d, "\GG h"] \\
\Bs \ar[r, "\beta"] & \GG \Bs
\end{tikzcd}
\]
We extend homomorphisms $h : \As \to \Bs$ to plays: 
\[ \hp([U_1, \ldots , U_n]) = [h(U_1), \ldots , h(U_n)]. \]
Since homomorphisms preserve guarded sets, this yields a well-defined map $\hp : \plays(\As) \to \plays(\Bs)$.
Given $\gd$-guarded decompositions $\tau : A \to \plays(\As)$ and $\upsilon : B \to \plays(B)$, a morphism of decompositions is a $\sg$-homomorphism $h : \As \to \Bs$ such that, for all $a \in A$, $\lb \hp(\tau(a)), h(a) \rb = \lb \upsilon(h(a)), h(a) \rb$. We say that a morphism is \emph{strict} if the stronger condition $\hp(\tau(a)) = \upsilon(h(a))$ holds.
\begin{proposition}
$\GG$-coalgebra morphisms are morphisms of the associated $\gd$-guarded decompositions, and vice versa. Moreover, injective coalgebra morphisms correspond to injective strict morphisms of decompositions.
\end{proposition}
\begin{proof}
Unpacking the diagram for coalgebra morphisms, it says that for all $a \in A$, if $\alpha(a) = \lb p,a \rb$, then $\beta(h(a)) = \lb \hp(p), h(a) \rb$. Since $\alpha(a) = \lb \tau(a),a \rb$ and $\beta(h(a)) = \lb \upsilon(h(a)), h(a) \rb$, the correspondence with morphisms of decompositions follows. In general, $\upsilon(h(a))$ may be a proper prefix of $\hp(\tau(a))$.
However, if $h$ is injective, this would imply that, for some proper prefix $p$ of $\tau(a)$, $a \in \lambda(p)$, which would contradict minimality for $\tau$.
\end{proof}
\section{Open maps and guarded bisimulation}

We now return to the issue of characterizing guarded bisimulation, and hence equivalence of structures modulo the full guarded fragments.
Following the approach taken in \cite{AbramskyShah2020}, the forthcoming extended journal version of \cite{DBLP:conf/csl/AbramskyS18}, we shall use a variant of the well-known open maps approach to bisimulation \cite{joyal1993bisimulation}.

We shall work in the Eilenberg-Moore category $\EMG$. Objects are $\GG$-coalgebras, and morphisms are $\GG$-coalgebra morphisms.
There is an evident forgetful functor $U : \EMG \to \CS$. This has a right adjoint $F$, which sends $\As$ to the cofree coalgebra $(\GG \As, \delta)$.

We say that a~$\gd$-guarded tree decomposition~$\tau$ is~\emph{simple} if $\Pt$ forms a covering chain $p_1 \cvr \cdots \cvr p_n$ under the prefix order.
We define \emph{paths} to be $\GG$-coalgebras whose associated decompositions $\tau$ are simple. If $P$ is a path with associated decomposition $\tau$, we identify $P$ with $\Pt$.
An \emph{embedding} is a coalgebra morphism which is an injective strong $\sg$-homomorphism, \ie it reflects as well as preserves the $\sg$-relations.
We denote embeddings by $\emb$.
\emph{Path embeddings} $e : P \emb \As$ are embeddings whose domain is a path. Note that, for all $p \in P$, $\lambda(p) \cong e(\lambda(p))$.

A $\sg$-homomorphism $h : \As \to \Bs$ is a \emph{pathwise embedding} if for each path embedding $e : P \emb \As$, $h \circ e$ is a path embedding.

We can now define what it means for a morphism $f : (\As, \alpha) \to (\Bs, \beta)$ in $\EMG$ to be \emph{open}.
This holds if, whenever we have a diagram
\[ \begin{tikzcd}
  P \arrow[r, rightarrowtail] \arrow[d,rightarrowtail]
    & Q \arrow[d, rightarrowtail] \\
  (\As, \alpha) \arrow[r,  "f"']
&(\Bs, \beta)
\end{tikzcd}
\]
where $P$ and $Q$ are paths, there is an embedding $Q \rightarrowtail (\As, \alpha)$ such that
\[ \begin{tikzcd}
  P \arrow[r, rightarrowtail] \arrow[d,rightarrowtail]
    & Q \arrow[dl, rightarrowtail] \arrow[d, rightarrowtail] \\
  (\As, \alpha) \arrow[r,  "f"']
&(\Bs, \beta)
\end{tikzcd}
\]
This is often referred to as the \emph{path-lifting property}. If we think of $f$ as witnessing a simulation of $\As$ by $\Bs$, path-lifting means that if we extend a given behaviour in $\Bs$ (expressed by extending the path $P$ to $Q$), then we can find a matching behaviour in $\As$ to ``cover'' this extension. Thus it expresses an abstract form of the notion of  ``p-morphism'' from modal logic \cite{blackburn2002modal}, or of functional bisimulation.

We can now define the back-and-forth equivalence $\As \lrarr^{\GG} \Bs$ between structures in $\CS$. This holds if there is a coalgebra $\Rs$ in $\EMG$, and a span of open pathwise embeddings
\[ \begin{tikzcd}
& \Rs \arrow[dl, rightarrowtail] \arrow[dr, rightarrowtail] \\
F\As & & F \Bs
\end{tikzcd}
\]

\begin{restatable}{theorem}{BisimulationSpanCorrespondence}
\label{thm:bisimulation-span-correspondence}
Let $\gd$ be a notion of guarding, and $\GG$ the corresponding comonad.
For all $\sg$-structures $\As$ and $\Bs$, $A \gbsim B$ iff $A \lrarr^{\GG} B$.
\end{restatable}
\begin{proof}
We begin by showing suitable spans yield Duplicator strategies. Assume we have a back and forth equivalence, so a coalgebra~$\Rs$ and span of open embeddings~$F \As \xleftarrow{p_1} \Rs \xrightarrow{p_2} \Bs$. 
We proceed by induction on the length of the play. At each stage we show:
\begin{enumerate}
\item that for every sequence of Spoiler moves and Duplicator responses of length~$n$ there exists a simple path~$\Ps$ in~$\Rs$ such that:
\begin{itemize}
    \item $p_1$ and~$p_2$ restrict to embeddings on~$\Ps$.
    \item The image under~$p_1$ and~$p_2$ are paths labelled with Spoiler moves and Duplicator responses.
\end{itemize}
\item Duplicator has a winning strategy for~$n$-rounds.
\end{enumerate}
For the base case, without loss of generality, assume Spoiler plays~$U_1$ in~$\As$. The set~$\{ \lb [U_1], u \rb \ \mid u \in U_1 \}$ is a simple path in~$F(\As)$. As~$p_1$ is open, there is a path~$\Ps_1$ in~$\Rs$ such that~$p_1$ restricts to an embedding onto the path in~$F(\As)$. As~$p_2$ is guarded set embedding, $p_2$ restricts to an embedding on~$\Ps_1$. By lemma~\ref{lem:guarded-set-properties}~$\ve$ restricts to an embedding on guarded sets. Composing these embeddings yields an isomorphism between the two sets in the image of the~$p_i$, witnessing a valid response for Duplicator for this round.

For the inductive step, assume the I.H. holds for~$n$. Without loss of generality, we only consider Spoiler moves in~$\As$. We also restrict to non-trivial Spoiler moves such that~$U_n \cap U_{n + 1} \neq \emptyset$, as other moves offer no advantage. By the I.H, we have a simple path of guarded sets~$U_1,\ldots,U_n$ in~$\As$ in the image of a path~$\Ps_n$ in~$\Rs$. We extend this path to a new simple path, by adding the set
\begin{equation*}
    \{ \lb [U_1,\ldots,U_n,U_{n + 1}], u \rb \mid u \in U_{n + 1}\}
\end{equation*}
The path is simple by the non-triviality of the Spoiler move. As~$p_1$ is open we can extend~$\Ps_n$ to a path~$\Ps_{n + 1}$, with~$p_1$ restricting to an embedding onto our extended path in~$F(\As)$. As~$p_2$ is guarded set embedding, and guarded sets cover all the relations in a path, it restricts to an embedding on~$\Ps_{n + 1}$. The final element of this path~$V_{n + 1}$ is our putative response for Duplicator. Restricting the embeddings and~$\ve$ yields the required isomorphism witnessing a valid response for Duplicator.

For the other direction, we aim to construct the intermediate ``diagonal'' coalgebra~$\langle Rs, \rho \rangle$, using an approach similar to that for the comonad construction~$\GG$. To do so, we consider triples:
\begin{equation*}
    \langle U, \iota, V \rangle
\end{equation*}
$U$ is a Spoiler move, and $\iota : U \rightarrow V$ is the response dictated by Duplicator's strategy, or~$V$ is a spoiler move and~$\iota^{-1} : V \rightarrow U$ is the response dictated by Duplicators strategy. A sequence~$p$ of such positions is the two-sided analog of the  plays seen earlier, and we pair these with a focus element:
\begin{equation*}
    \langle p, \langle a, b \rangle \rangle
\end{equation*}
If~$\langle U, \iota, V \rangle$ is the last element of~$p$, then we require~$a \in U$,$b \in V$ and~$\iota(a) = b$.Clearly there is some redundancy in this description, but it is convenient to manipulate. 
Finally we quotient, identifying~$\langle q, \langle a, b \rangle \rangle$ with~$\langle q, \langle a, b \rangle \rangle$ if~$p$ and~$q$ have a common prefix, and~$\langle a, b \rangle$ appears consistently on the paths from the common prefix to the endpoints. 
Relations are defined in a similar manner to those for~$\GG$. We define~$\rho : Rs \rightarrow \GG(\Rs)$ as:
\begin{align*}
    \rho &\lb [\langle U_1, \iota_1, V_1 \rangle,...,\langle U_n, \iota_n, V_n \rangle], \langle a, b \rangle \rb\\ 
     &\,=\,\lb [W_1,...,W_n], \lb [\langle U_1, \iota_1, V_1 \rangle,..., \langle U_n, \iota_n, V_n \rangle], \langle a, b \rangle \rb \rb
\end{align*}
where
\begin{equation*}
    W_j = \{ \lb [\langle U_1,\iota_1,V_1 \rangle,..., \langle U_j, \iota_j, V_j \rangle], \langle a, \iota_j a \rangle \rb \mid a \in U_j \}
\end{equation*}
The function $\rho$ is a homomorphism, and in fact~$(R, \rho)$ an Eilenberg-Moore coalgebra. 
Further, there is a map:
\begin{equation*}
    p_1(\lb p, \langle a,  b \rangle \rb) = a
\end{equation*}
and similarly for~$p_2$ onto the second component. The~$p_i$ are coalgebra homomorphisms, open, and guarded set embedding, completing the proof.
\end{proof}

Again, for atom and loose guards, there is also a resource-bounded version of this result, where we work in $\EMGk$ rather than $\EMG$. We can then define the equivalence $\As \lrarr^{\Gk} \Bs$ in exactly analogous fashion to the above.
\begin{theorem}
If~$\gd$ is either atom or loose guards,
for $k>0$, and resource bounded comonad~$\Gk$,
for all $\sg$-structures $\As$ and $\Bs$, $A \gbsimk B$ iff $A \lrarr^{\Gk} B$.
\end{theorem}
The forthcoming~\cite{AbramskyShah2020} presents an axiomatic framework for model comparison games, exemplified by results for Ehrenfeucht-\Fraisse, pebbling and modal bisimulation games. Although developed concretely, our results such as theorem~\ref{thm:bisimulation-span-correspondence} can be viewed from this perspective, as significantly more elaborate instances of the axiomatic framework. Our notion of guarded set preservation corresponds to the pathwise embedding condition used in the abstract setting.

\section{The connection to hypergraphs}

The guarded comonad constructions  have an underlying combinatorial content which is largely independent of the specifics of relational $\sg$-structures, or the syntactic form of guarding being considered. Semantically, what matters is that we have some designated family of finite subsets of the universe (the ``guarded sets''), and that morphisms preserve these sets.
This shift in perspective has been advocated in~\cite{HodkinsonOtto2003}. We can take advantage of our structural approach to develop this idea more fully.

We will show that the guarded comonad construction makes sense at the level of hypergraphs. Moreover, the guarded comonads for $\sg$-structures are nicely related to this hypergraph comonad, by an Eilenberg-Moore law \cite{manes2007monad,jacobs2017introduction}.

A \emph{hypergraph}~$\hyp{}$ is given by a set of \emph{vertices}~$\hver{}$ and a family $\hedg{}$  of finite subsets of $\hver{}$. The elements of $\hedg{}$ are referred to as~\emph{hyperedges}.
A \emph{morphism of hypergraphs}~$f : \hyp{1} \rightarrow \hyp{2}$ is a function on the vertices~$f : \hver{1} \rightarrow \hver{2}$ mapping hyperedges in~$\hedg{1}$ to hyperedges in~$\hedg{2}$. We write $\hgraph$ for the category of hypergraphs.

We extend the notions of simulation and bisimulation to hypergraphs.
For hypergraphs $\hyp{1}$, $\hyp{2}$, the hypergraph simulation game from $\hyp{1}$ to $\hyp{2}$
is described as follows:
\begin{itemize}
\item Round 0: We set $X_0 := \vempty$, $\vphi_0 := \vempty$.
\item Round $n+1$: Spoiler specifies a hyperedge $X_{n+1}$ in $\hedg{1}$. Duplicator must respond with a hyperedge $Y_{n+1}$ in $\hedg{2}$, and a function $\vphi_{n+1} : X_{n+1} \to Y_{n+1}$, such that $\vphi_{n+1} |_{X} = \vphi_{n} |_{X}$, where $X = X_{n+1} \cap X_{n}$.
\end{itemize}
Duplicator wins if he has a response at each round, otherwise Spoiler wins. 

As with the previous notions of bisimulation, the hypergraph bisimulation game is the two sided generalization of the simulation game:
\begin{itemize}
\item Round 0: We set $X_0 := \vempty$, $Y_0 := \vempty$, $\vphi_0 := \vempty$.
\item Round $n+1$: Spoiler now has two options.
\begin{itemize}
\item Option 1: Spoiler specifies a hyperedge $X_{n+1}$ in $\hedg{1}$. Duplicator must respond with a hyperedge $Y_{n+1}$ in $\hedg{2}$, and a bijection $\vphi_{n+1} : X_{n+1} \to Y_{n+1}$, such that $\vphi_{n+1} |_{X} = \vphi_{n} |_{X}$, where $X = X_{n+1} \cap X_{n}$.
\item Option 2: Spoiler specifies a hyperedge $Y_{n+1}$ in $\hyp{2}$. Duplicator must respond with a hyperedge $X_{n+1}$ in $\hyp{1}$, and a bijection $\vphi_{n+1} : X_{n+1} \to Y_{n+1}$, such that $\vphi_{n+1}^{-1} |_{Y} = \vphi_{n}^{-1} |_{Y}$, where $Y = Y_{n+1} \cap Y_{n}$.
\end{itemize}
\end{itemize}

Again, Duplicator wins if he has a response at each round, otherwise Spoiler wins. There are resource bounded variants of these games, in which moves are restricted to hyperedges of at most~$k$ elements.

To provide a structural account of these games, we proceed as before by encoding Spoiler's moves within a hypergraph. We consider \emph{focussed plays} $\langle p, a \rangle$ where~$p$ is a path consisting of a non-empty list~$[U_1,...U_n]$ of hyperedges, and~$a \in U_n$ is a focus. 
We define an equivalence relation on focussed plays to enforce the overlap condition in the simulation game, in entirely analogous fashion to that of section~\ref{sec:guarded-comonads}. We denote the equivalence class of $\lt p,a \rt$ by~$\lb p, a \rb$. We then construct a new hypergraph, $\HH \hyp{}$. The vertices are the equivalence classes $\lb p,a \rb$.
The hyperedges are the sets of the form:
$\{ \lb p, a \rb \mid a \in U \}$,
where~$U \in \hedg{}$.

\begin{restatable}{theorem}{HypergraphComonad}
There is a comonad in Kleisli form~$(\HH, \ve, (\cdot)^*)$ on $\hgraph$. The counit $\ve$ is defined on representatives by~$\ve \lb p, a \rb = a$. Coextension is defined on representatives by
\begin{equation*}
    h^* \lb [U_1,...,U_n], a \rb = \lb [V_1,...,V_n], h \lb p, a \rb \rb
\end{equation*}
where
\begin{equation*}
    V_j = \{ h \lb [U_1,...,U_j], a \rb \mid a \in U_j \}, \quad 1 \leq j \leq n.
\end{equation*}
\end{restatable}
\begin{proof}
Let~$(V,E)$ be a hypergraph, and $\HH(V,E)$ the resulting hypergraph yielded from the proposed construction. Firstly, we note that~$\ve$ is a well defined function as equivalence classes agree on their second component. Now assume that $\{\lb p_u, u \rb \mid u \in U \}$ is a hyperedge in~$\HH(V,E)$. There can only be such a hyperedge if there exists a~$p$ such that for all~$u \in U$:
\begin{equation*}
    \lb p_u, u \rb = \lb p, u \rb
\end{equation*}
and
\begin{equation*}
    U \in E
\end{equation*}
and so~$\ve$ preserves hyperedges, and is therefore a homomorphism of hypergraphs.

Assume that~$h : \HH(V_1,E_1) \rightarrow (V_2,E_2)$ is a hypergraph morphism. We first aim to show that~$h^* : \HH(V_1,E_1) \rightarrow \HH(V_2,E_2)$ is a well defined homomorphism. 

We first note that the sets~$V_j$
\begin{equation*}
    V_j = \{ h \lb [U_1,\ldots,U_j], u \rb \mid u \in U_j \}
\end{equation*}
used in the definition of~$h^*$ are hyperedges in~$(V_2,E_2)$. This is simply because sets of the form~$\{ \lb [U_1,\ldots,U_j], u \rb \mid u \in U_j \}$ are hyperedges in~$\HH(V_1,E_1)$ by definition, and~$h$ is a homomorphism of type~$\HH(V_1,E_1) \rightarrow (V_2,E_2)$.

To establish that~$h^*$ is well defined with respect to equivalence classes, assume:
\begin{equation*}
    \lb p, u \rb = \lb q, v \rb
\end{equation*}
By the definition of the equivalence relation, $u = v$, $p \sqcap q$ exists, and
\begin{equation*}
    \lb p, u \rb = \lb p \sqcap q, \rb = \lb q, u \rb
\end{equation*}
Furthermore, $u$ appears in all the hyperedges in the extension of~$p \sqcap q$ to both~$p$ and~$q$. Let:
\begin{equation*}
    p = [U_1,\ldots,U_n]
\end{equation*}
with prefix
\begin{equation*}
    p \sqcap q = [U_1,\ldots,U_m]
\end{equation*}
Then, with the~$V_j$ as defined in the theorem statement:
\begin{equation*}
    [V_1,\ldots,V_m] \preford [V_1,\ldots,V_n]
\end{equation*}
with~$h(p,u)$ appearing in all the hyperedges~$V_m,\ldots,V_n$. Therefore:
\begin{equation*}
    h^* \lb p, u \rb = h^* \lb p \sqcap q, u \rb
\end{equation*}
and by a dual argument
\begin{equation*}
    h^* \lb q, u \rb = h^* \lb p \sqcap q, u \rb
\end{equation*}
and transitivity completes this part of the proof.

To show~$h^*$ preserves hyperedges, assume~$\{ \lb p_u, u \rb \mid u \in U \}$ is a hyperedge in~$\HH(V_1,E_1)$. There must exist~$p$ such that for all~$u \in U$
\begin{equation*}
    \lb p_u, u \rb = \lb p, u \rb
\end{equation*}
and
\begin{equation*}
    \{ h \lb p, u \rb \mid u \in U \}
\end{equation*}
is a hyperedge in~$(V_2,E_2)$. By definition, there exist~$q$ such that
\begin{equation*}
    h^* \{ \lb p, u \rb \mid u \in U \}  = \{ \lb q, h \lb p, u \rb \rb \mid u \in U \} 
\end{equation*}
The right hand side is equal to
\begin{equation*}
    \{ \lb q, v \rb \mid v \in \{ h \lb p, u \rb \mid u \in U \} \}
\end{equation*}
and so, by the definition of hyperedge in~$\HH(V_2,E_2)$, $h^*$ preserves hyperedges.

Finally, we must verify the comonad axioms.
Firstly we aim to show:
\begin{equation*}
  \ve^*_{\As} = id_{\HH\As}  
\end{equation*}
Consider~$\lb [U_1,\ldots,U_n], a \rb \in \HH \As$. There exists~$[V_1,\ldots,V_n]$ such that
\begin{align*}
    \ve_{\As}^*(\lb [U_1,\ldots,U_n], a \rb) &= \lb [V_1,\ldots,V_n], \ve_{\As}(\lb [U_1, \ldots, U_n], a \rb) \rb\\ &= \lb [V_1, \ldots, V_n], a \rb
\end{align*}
It is therefore sufficient to show that $U_j = V_j$ for~$1 \leq j \leq n$. We have:
\begin{equation*}
    V_j = \{ \epsilon \lb [U_1,\ldots,U_j], a_j \rb \mid a_j \in U_j \} = \{ a_j \mid a_j \in U_j \} = U_j
\end{equation*}
completing this part of the proof.

Next, we aim to show, for all~$f : \HH(\As) \rightarrow \Bs$:
\begin{equation*}
   \ve \circ f^* = f 
\end{equation*}
For~$\lb p, a \rb \in \HH \As$, there exists~$q$ such that
\begin{equation*}
    \ve_{\As} \circ f^*(\lb p, a \rb) = \ve_{\As}(\lb q, f(\lb p, a \rb) \rb) = f(\lb p, a \rb)
\end{equation*}

It remains to confirm the third axiom, for all
\begin{equation*}
f : \HH(\As) \rightarrow \Bs \quad\mbox{ and }\quad g : \HH(\Bs) \rightarrow \Cs
\end{equation*}
the following distributivity condition holds
\begin{equation*}
(g \circ f^*)^* = g^* \circ f^*    
\end{equation*}
Again consider~$\lb [U_1,\ldots,U_n], a \rb \in \HH \As$. There exist
\begin{equation*}
    V_1,\ldots,V_n \qquad W_1,\ldots,W_n \qquad X_1,\ldots,X_n
\end{equation*}
such that:
\begin{align*}
    &g^* \circ f^* \lb [U_1,\ldots,U_n], a \rb \\ &= g^* \lb [V_1,\ldots,V_n], f(\lb [U_1,\ldots,U_n], a \rb) \rb\\
    &= \lb [W_1,\ldots,W_n], g(\lb [V_1,\ldots,V_n], f(\lb [U_1,\ldots, U_n], a \rb) \rb) \rb
\end{align*}
and
\begin{align*}
    &(g \circ f^*)^* (\lb [U_1,\ldots,U_n], a \rb) \\ &= \lb [X_1,\ldots,X_n], g \circ f^*(\lb [U_1,\ldots,U_n], a\rb) \rb\\
    &= \lb [X_1,\ldots, X_n], g(\lb [V_1,\ldots,V_n], f(\lb [U_1,\ldots,U_n], a\rb) \rb) \rb
\end{align*}
Therefore, it is sufficient to show
\begin{equation*}
    [W_1,\ldots,W_n] = [X_1,\ldots,X_n]
\end{equation*}
For index~$1 \leq j \leq n$
\begin{align*}
    X_j &= \{ g \circ f^* (\lb [U_1,\ldots,U_j], a_u \rb) \mid a_u \in U_j \}\\
    &= \{ g(\lb [V_1,\ldots,V_j], f(\lb [U_1,\ldots,U_j], a_u \rb)\rb) \mid a_u \in U_j \}
\end{align*}
and
\begin{align*}
    W_j &= \{ g(\lb [V_1,\ldots,V_j], a_v \rb)  \mid a_v \in V_j \}\\
    &= \{ g(\lb [V_1,\ldots,V_j], f(\lb [U_1,\ldots,U_j], a_u \rb) \rb) \mid a_u \in U_j \}
\end{align*}
Therefore the lists have the same length and are equal pointwise.

\end{proof}

\textbf{Remark} The above arguments are formally almost identical to those for the guarded bisimulation comonads~$\Gg$. Slight modifications are required to replace relation preservation with hyperedge preservation. The details are conceptually clearer and more uniform than for the various notions of guarding. The notions of related tuples and guarded sets merge into the single concept of hyperedges.

This comonad restricts to resource bounded versions.
\begin{theorem}
For each~$k > 0$ there is a comonad in Kleisli form~$(\Hk, \ve, (\cdot)^*)$ on $\hgraph$, given by restricting focussed plays to hyperedges of size bounded by~$k$. 
\end{theorem}
We now relate this construction to the guarded comonads for relational structures. For each notion of guarding $\gd$, there is a functor $H^{\gd} : \CS \to \hgraph$, which sends $\As$ to the hypergraph $(A, E)$, where $E$ is the set of $\gd$-guarded subsets of $\As$. 
\begin{restatable}{theorem}{EMLaw}
\label{thm:em-law}
For each notion of guarding $\gd$, there is a natural isomorphism
\[ \theta : H^{\gd} \circ \GG^{\gd} \; \cong \; \HH \circ H^{\gd} \]
such that, for all $\As$ in $\CS$, and omitting $\gd$ to declutter notation, the following diagrams commute:
\[ \begin{tikzcd}[ampersand replacement=\&]
H \GG A \ar[d, "\theta_{\As}"'] \ar[r, "H \ve_{\As}"] \& H \As \ar[dl, "\ve_{H \As}"] \\
\HH H \As
\end{tikzcd}
\quad
\begin{tikzcd}[ampersand replacement=\&]
H\GG\As \ar[d, "\theta_{\As}"'] \ar[r, "H\delta_{\As}"] \& H\GG\GG\As \ar[r, "\theta_{\GG \As}"] \& \HH H \GG \As \ar[dl, "\HH \theta_{\As}"] \\
\HH H \As \ar[r, "\delta_{H \As}"'] \&  \HH \HH H \As
\end{tikzcd}
\]
Thus $\theta$ forms an Eilenberg-Moore law (see e.g.~\cite{jacobs2017introduction}).
\end{restatable}
\begin{proof}
Assume~$h : \As \rightarrow \Bs$, and~$U$ is a hyperedge of~$H^{\gd}(\As)$. There must be a tuple of elements~$\overline{a}$, with support containing~$U$, and guard formula~$\varphi$ such that~$\varphi^\As(\overline{a})$ holds. As all the various guard formulae are existential positive, this implies~$\varphi^\Bs(h\overline{a})$, and so~$h(U)$ is a hyperedge of~$H^{\gd}(\Bs)$. So~$H^\gd(h)$ is a valid hypergraph morphism. Composition and identities are trivially preserved, so the construction is functorial.

That the triangle and pentagon commute is straightforward, as both counits and comultiplications~``do the same thing'' at the level of sets.
\end{proof}
In fact,  $\theta$ is simply the identity, showing the close relationship of these constructions.

In the atom and loose guarded cases, there is a corresponding version for the resource-bounded comonads, for functors~$H^\gd_k$ restricting to~$k$-guarded sets.
\begin{restatable}{theorem}{GradedEMLaw}
If~$\gd$ is either atom or loose guards, there is a natural isomorphism
\[ \theta : H^{\gd}_k \circ \Gk^{\gd} \; \cong \; \Hk \circ H^{\gd}_k \]
which commutes with the counits and comultiplications of $\Gk$ and $\Hk$.
Thus $\theta$ forms an Eilenberg-Moore law.
\end{restatable}
\begin{proof}
This is a straightforward restriction of the proof for theorem~\ref{thm:em-law}. Note that we need to use the restricted functors~$H^\gd_k$ to be consistent with the bookkeeping of guarded sets and hyperedges. This is because the guarded set comonad closes under subsets, but the hypergraph comonad does not so we must allow for this difference.
\end{proof}
We can also develop characterizations of the coalgebras for $\HH$ and $\Hk$, and of hypergraph bisimulation in terms of spans of open maps, along entirely similar lines to what we have done for $\GG$ and $\Gk$.
The characterization of hypergraph coalgebras in terms of guarded tree decompositions shows that hypergraph coalgebras can be seen as \emph{acyclicity witnesses}. There are interesting connections to the work of Otto \cite{Otto2012a,Otto2020a,Otto2020b}, which we hope to pursue in future work.

\section{A Second Grading}
\label{sec:depth-bounds}
We have considered a grading based on restricting the size of guarded sets used in the bisimulation game, which corresponds to restricting the number of variables used in guards.
Another natural restriction is to bound the number of rounds of the game, and this parameter also has a clear logical interpretation. Extending our previous notation, we write~$\As \gsimd \Bs$ if Duplicator has a winning strategy for $d$ rounds of the~$\gd$-guarded simulation game, and~$\As \gbsimd \Bs$ for the corresponding notion for the full bisimulation game.

We can then further strengthen the relationships of theorem~\ref{bisimlogth} with the following standard results~\cite{gradel2014freedoms} relating logic and game duration.
\begin{theorem}
\label{thm:depth-bounded}
\begin{enumerate}
\item $\As \gbsimd \Bs$ iff $\As$ and $\Bs$ satisfy the same $\gd$-guarded formulas of guarded quantifier depth $\leq d$.
\item $\As \gsimd \Bs$ iff every existential positive $\gd$-guarded formula of guarded quantifier depth $\leq d$ satisfied by $\As$ is satisfied by $\Bs$.
\end{enumerate}
\end{theorem}
Note that as guarded quantifiers occur over tuples of variables, the guarded-quantifier depth will typically be lower than the conventional quantifier depth in the sense of first-order logic.

This resource parameter can also be incorporated into our framework.
As with the width grading considered in earlier sections, our constructions can only be bounded in the atom and loosely guarded cases. Again, this is because of the existential quantifiers present in clique guards. For similar reasons, we must restrict our attention to a comonad where  plays only consist of exact guarded sets.
\begin{theorem}
\label{thm:exactly-guarded-comonad}
The construction of theorem~\ref{thm:guarded-comonad} can be restricted to only consider exactly guarded sets in  plays, yielding a comonad~$(\Geg,\ve,(\cdot)^*)$.
\end{theorem}
These variants will capture the same logical fragment, as it is sufficient to consider exactly guarded sets in the model comparison games. In the atom and loose guarded cases, we can then restrict to yield comonads bounded by play length.
\begin{theorem}
If~$\gd$ is either atom or loose guards and~$d > 0$, there is a comonad $\Gdg$, given by restricting the construction of theorem~\ref{thm:exactly-guarded-comonad} to  plays of length at most~$d$. Furthermore, for all $\sg$-structures~$\As$, $\Bs$:
\begin{enumerate}
    \item There is a bijective correspondence between:
    \begin{enumerate}
        \item CoKleisli morphisms~$\Gdg(\As) \rightarrow \Bs$.
        \item Winning strategies for Duplicator in the $d$-round $\gd$-guarded simulation game from~$\As$ to~$\Bs$.
    \end{enumerate}
    \item $\As \gbsimd \Bs$ iff $\As \lrarr^{\Gdg} \Bs$.
\end{enumerate}
\end{theorem}
\begin{proof}
The arguments are formally identical to the previous theorems, once we note that in the atom and loose guarded cases we made no essential use of lengthening  plays in our arguments.
\end{proof}
These variants introduce new combinatorial parameters, for coverings by exactly guarded sets, and bounding the depth of the forest. As suggested by our notation, in the general case we can bound on both parameters together, yielding comonads~$\Gkdg$ restricting both guarded set width and game duration, simultaneously exploiting the logical correspondences of theorems~\ref{bisimlogth} and~\ref{thm:depth-bounded}.
This extension is routine, so we will not belabour the details.

\section{Atoms and Cliques}
\label{sec:atoms-and-cliques}
We have observed that for $\gd$ either atom or loose guards, the comonad~$\Gg$ restricts to a comonad~$\Gkg$, corresponding to  restricting plays of the bisimulation game of width bound~$k$. Unfortunately, due to the witnesses for the existential quantifiers appearing in clique guards, such a restriction cannot be directly applied in that case. 

To address this, for a signature~$\sg$ we introduce an extended signature~$\sg^*$ with additional relation symbols~$C_i$, with~$i$ ranging over the  positive natural numbers. With this in place, we can expose Gaifman cliques as relational atoms.
\begin{lemma}
For any~$\sg$-structure $\As$ we define a~$\sg^*$-structure~$C(\As)$, extending~$\As$, with:
\begin{equation*}
    C^{C(\As)}_n(a_1,\ldots,a_n) \quad\Leftrightarrow\quad \As \models \cliquen{n}(a_1,\ldots,a_n)
\end{equation*}
This mapping is the object part of a functor $\CStruct{\sg} \rightarrow \CStruct{\sg^*}$. The atom guarded sets of~$C(\As)$ are exactly the clique guarded sets of~$\As$.
\end{lemma}
\begin{proof}
Functoriality follows from the~$C_i$ being defined by existential positive formulae, and hence preserved by homomorphisms. The second claim is immediate from the construction.
\end{proof}
This assignment allows us to functorially reduce clique bisimilarity to the atom guarded case.
\begin{proposition}
For~$\sg$-structures~$\As$ and~$\Bs$ the following are equivalent:
\begin{enumerate}
    \item \label{en:clique-bisim} Duplicator has a winning strategy for the clique guarded bisimulation game between~$\As$ and~$\Bs$.
    \item \label{en:atom-bisim} Duplicator has a winning strategy for the atom guarded bisimulation game between~$C(\As)$ and~$C(\Bs)$.
\end{enumerate}
Further, this remains true if we restrict the number of rounds played, the width of guards, or if we restrict to the existential variant of the game.
\end{proposition}
\begin{proof}
To show~\ref{en:clique-bisim} implies~\ref{en:atom-bisim}, we proceed by induction. For the base case, assume without loss of generality that Spoiler plays~$\{ a_1, \ldots, a_n \} \subseteq C(\As)$ in the atom guarded game. We must have~$C_n^{C(\As)}(a_1, \ldots, a_n)$, and so~$\As \models \cliquen{n}(a_1,\ldots, a_n)$. As we have a winning strategy for the clique guarded game, there must exist a set~$\{ b_1, \ldots, b_n \} \subseteq \Bs$ such that the mapping~$a_i \mapsto b_i$ is a partial isomorphism, and~$\Bs \models \cliquen{n}(b_1,\ldots,b_n)$. We take this mapping as Duplicator's response. Clearly this mapping preserves relations in~$\sg$. For the~$C_i$ we simply note that induced sub-graphs of cliques are cliques. Finally, the inductive step proceeds almost identically, with the required agreement on overlaps between the partial isomorphisms following from the same condition in the clique guarded game.

To show~\ref{en:atom-bisim} implies~\ref{en:clique-bisim}, again we proceed by induction. For the base case, assume without loss of generality that Spoiler plays~$\{ a_1, \ldots, a_n \}$, and so~$\As \models \cliquen{n}(a_1,\ldots, a_n)$, implying~$C_n^{C(\As)}(a_1,\ldots,a_n)$. By assumption, there is an atom guarded set~$\{ b_1, \ldots b_n \}$ such that the mapping~$a_i \mapsto b_i$ is a partial isomorphism. We take this as Duplicator's response. As partial isomorphisms preserve atoms, we must have~$C_n^{C(\Bs)}(b_1,\ldots,b_n)$. and so~$\cliquen{n}(b_1,\ldots,b_n)$, and this restricts to a partial isomorphism between guarded sets over~$\sg$. Again, the inductive case proceeds almost identically, with the overlap condition for the partial isomorphisms in each round being inherited from the rules of the atom guarded game.

The restriction to bounded play length or guard width is routine. The simulation game claim can be seen by restricting Spoiler and Duplicator plays to the required components, and noting that moving to homomorphisms rather than partial isomorphisms causes no trouble in the argument.
\end{proof}
Using this result, resource-bounded clique guarded bisimilarity can be captured by our comonadic framework. The signature extension and corresponding model transformation functor can be restricted to only encode cliques below a size bound. In this way we can remain within the realm of finite signatures, whilst capturing clique guarded bisimulation for all structures with universe cardinality below a specified finite bound.

\section{Final remarks}

There are many further directions arising from this work:
\begin{itemize}
    \item There has been an extensive development of negation-guarded fragments \cite{barany2012queries,benedikt2015complexity,danielski2018unary}. We will give a comonadic account of these in a sequel to this paper.
    \item We will also study the connections with acyclic covers and the work of Otto in \cite{Otto2012a,Otto2020a,Otto2020b} from our coalgebraic perspective.
    \item Another promising line of investigation is to look at the tree model properties of guarded logics, and the associated decision procedures and complexity results, from the comonadic point of view. 
\end{itemize}

\section*{Acknowledgements} 
We thank the participants in the Comonads seminar associated with the EPSRC project EP/T00696X/1 ``Resources and Co-Resources: a junction between categorical semantics and descriptive complexity'' for their encouragement and feedback. EPSRC support is gratefully acknowledged.

\bibliography{bibfile}

\begin{thebibliography}{10}
\providecommand{\url}[1]{#1}
\csname url@samestyle\endcsname
\providecommand{\newblock}{\relax}
\providecommand{\bibinfo}[2]{#2}
\providecommand{\BIBentrySTDinterwordspacing}{\spaceskip=0pt\relax}
\providecommand{\BIBentryALTinterwordstretchfactor}{4}
\providecommand{\BIBentryALTinterwordspacing}{\spaceskip=\fontdimen2\font plus
\BIBentryALTinterwordstretchfactor\fontdimen3\font minus
  \fontdimen4\font\relax}
\providecommand{\BIBforeignlanguage}[2]{{%
\expandafter\ifx\csname l@#1\endcsname\relax
\typeout{** WARNING: IEEEtran.bst: No hyphenation pattern has been}%
\typeout{** loaded for the language `#1'. Using the pattern for}%
\typeout{** the default language instead.}%
\else
\language=\csname l@#1\endcsname
\fi
#2}}
\providecommand{\BIBdecl}{\relax}
\BIBdecl

\bibitem{abramsky2017pebbling}
S.~Abramsky, A.~Dawar, and P.~Wang, ``The pebbling comonad in finite model
  theory,'' in \emph{Logic in Computer Science (LICS), 2017 32nd Annual
  ACM/IEEE Symposium on}.\hskip 1em plus 0.5em minus 0.4em\relax IEEE, 2017,
  pp. 1--12.

\bibitem{DBLP:conf/csl/AbramskyS18}
S.~Abramsky and N.~Shah, ``Relating structure and power: Comonadic semantics
  for computational resources,'' in \emph{27th {EACSL} Annual Conference on
  Computer Science Logic, {CSL} 2018, September 4-7, 2018, Birmingham, {UK}},
  2018, pp. 2:1--2:17.

\bibitem{AbramskyShah2020}
------, ``Relating structure and power: Comonadic semantics for computational
  resources,'' Extended version to appear in {Journal of Logic and
  Computation}. Preprint available at \url{https://arxiv.org/abs/2010.06496},
  2021.

\bibitem{andreka1998modal}
H.~Andr{\'e}ka, I.~N{\'e}meti, and J.~van Benthem, ``Modal languages and
  bounded fragments of predicate logic,'' \emph{Journal of Philosophical
  Logic}, vol.~27, no.~3, pp. 217--274, 1998.

\bibitem{gradel1999decision}
E.~Gr{\"a}del, ``Decision procedures for guarded logics,'' in
  \emph{International Conference on Automated Deduction}.\hskip 1em plus 0.5em
  minus 0.4em\relax Springer, 1999, pp. 31--51.

\bibitem{gradel2014freedoms}
E.~Gr{\"a}del and M.~Otto, ``The freedoms of (guarded) bisimulation,'' in
  \emph{Johan van Benthem on Logic and Information Dynamics}.\hskip 1em plus
  0.5em minus 0.4em\relax Springer, 2014, pp. 3--31.

\bibitem{HodkinsonOtto2003}
I.~Hodkinson and M.~Otto, ``Finite conformal hypergraph covers and {G}aifman
  cliques in finite structures,'' \emph{Bulletin of Symbolic Logic}, vol.~9,
  no.~3, pp. 387--405, 2003.

\bibitem{pierce1991basic}
B.~C. Pierce, \emph{Basic category theory for computer scientists}.\hskip 1em
  plus 0.5em minus 0.4em\relax MIT press, 1991.

\bibitem{abramsky2010introduction}
S.~Abramsky and N.~Tzevelekos, ``Introduction to categories and categorical
  logic,'' in \emph{New structures for physics}.\hskip 1em plus 0.5em minus
  0.4em\relax Springer, 2010, pp. 3--94.

\bibitem{marx2001tolerance}
M.~Marx, ``Tolerance logic,'' \emph{Journal of Logic, Language and
  Information}, vol.~10, no.~3, pp. 353--374, 2001.

\bibitem{blackburn2002modal}
P.~Blackburn, M.~De~Rijke, and Y.~Venema, \emph{Modal Logic}.\hskip 1em plus
  0.5em minus 0.4em\relax Cambridge University Press, 2002, vol.~53.

\bibitem{van1977modal}
J.~F. A.~K. van Benthem, ``Modal correspondence theory,'' Ph.D. dissertation,
  University of Amsterdam, 1976.

\bibitem{JaninWalukiewicz1996}
D.~Janin and I.~Walukiewicz, ``On the expressive completeness of the
  propositional mu-calculus with respect to monadic second order logic,'' in
  \emph{International Conference on Concurrency Theory}.\hskip 1em plus 0.5em
  minus 0.4em\relax Springer, 1996, pp. 263--277.

\bibitem{GradelHirschOtto2002}
E.~Gr{\"a}del, C.~Hirsch, and M.~Otto, ``Back and forth between guarded and
  modal logics,'' \emph{ACM Transactions on Computational Logic (TOCL)},
  vol.~3, no.~3, pp. 418--463, 2002.

\bibitem{manes2012algebraic}
E.~G. Manes, \emph{Algebraic Theories}.\hskip 1em plus 0.5em minus 0.4em\relax
  Springer Science \& Business Media, 2012, vol.~26.

\bibitem{joyal1993bisimulation}
A.~Joyal, M.~Nielson, and G.~Winskel, ``Bisimulation and open maps,'' in
  \emph{Proceedings Eighth Annual IEEE Symposium on Logic in Computer
  Science}.\hskip 1em plus 0.5em minus 0.4em\relax IEEE, 1993, pp. 418--427.

\bibitem{manes2007monad}
E.~G. Manes and P.~Mulry, ``{Monad compositions. I: General constructions and
  recursive distributive laws.}'' \emph{Theory and Applications of Categories
  [electronic only]}, vol.~18, pp. 172--208, 2007.

\bibitem{jacobs2017introduction}
B.~Jacobs, \emph{Introduction to Coalgebra}.\hskip 1em plus 0.5em minus
  0.4em\relax Cambridge University Press, 2017, vol.~59.

\bibitem{Otto2012a}
M.~Otto, ``Highly acyclic groups, hypergraph covers, and the guarded
  fragment,'' \emph{Journal of the ACM (JACM)}, vol.~59, no.~1, pp. 1--40,
  2012.

\bibitem{Otto2020a}
------, ``Amalgamation and symmetry: From local to global consistency in the
  finite,'' \emph{arXiv preprint arXiv:1709.00031}, 2020.

\bibitem{Otto2020b}
------, ``Acyclicity in finite groups and groupoids,'' \emph{arXiv preprint
  arXiv:1806.08664}, 2020.

\bibitem{barany2012queries}
V.~Barany, B.~t. Cate, and M.~Otto, ``Queries with guarded negation (full
  version),'' \emph{arXiv preprint arXiv:1203.0077}, 2012.

\bibitem{benedikt2015complexity}
M.~Benedikt, B.~Ten~Cate, T.~Colcombet, and M.~V. Boom, ``The complexity of
  boundedness for guarded logics,'' in \emph{2015 30th Annual ACM/IEEE
  Symposium on Logic in Computer Science}.\hskip 1em plus 0.5em minus
  0.4em\relax IEEE, 2015, pp. 293--304.

\bibitem{danielski2018unary}
D.~Danielski and E.~Kiero{\'n}ski, ``Unary negation fragment with equivalence
  relations has the finite model property,'' in \emph{Proceedings of the 33rd
  Annual ACM/IEEE Symposium on Logic in Computer Science}, 2018, pp. 285--294.

\end{thebibliography}

\end{document}